\documentclass[conference]{IEEEtran}
\usepackage{geometry}
\usepackage[ruled,vlined,linesnumbered,noresetcount]{algorithm2e}
\usepackage{stmaryrd}
\usepackage{graphicx}
\usepackage{amsthm}
\usepackage{empheq}
\usepackage{url}
\usepackage{color,mathtools}
\usepackage{geometry}
\usepackage{setspace}
\usepackage{amsmath}
\usepackage{amsthm}
\usepackage{array}
\usepackage{gensymb}
\usepackage{amssymb}
\usepackage{url}
\usepackage{multicol}
\usepackage{multirow}
\usepackage{framed}
\usepackage{xcolor}
\usepackage{caption}
\usepackage{verbatim}
\usepackage{subfig}
\usepackage{amsfonts}
\usepackage{colortbl}
\usepackage{wrapfig}
\usepackage{lipsum}
\usepackage{enumerate}
\usepackage{epstopdf}

\usepackage{tikz}
\newcommand{\myarrow}[1][-45]{%
  \mathrel{%
    \text{$
     \begin{tikzpicture}[baseline = -0.5ex]
       \node[inner sep=0pt,outer sep=0pt,rotate = #1] (a) at (0,0)  {$\xrightarrow{}$};
    \end{tikzpicture}
    $}%
  }%
}%

\DeclareTextFontCommand{\code}{\fontfamily{pcr}\selectfont}
\DeclareTextFontCommand{\file}{\fontfamily{cmss}\selectfont}
\newtheorem{theorem}{Theorem}[section]

\newtheorem{lemma}{Lemma}[section]

\makeatletter
\def\blfootnote{\gdef\@thefnmark{}\@footnotetext}
\makeatother

\def\alon#1{{\sc Alon says: }{\color{red}\sf**  #1 ** }} 

\def\jie#1{{\sc Jie says: }{\color{magenta}\sf #1}}

\usepackage{graphics}

\def\alon#1{{}}\def\jie#1{{}}
\newcommand{\reals}{\mathbb{R}}
\newcommand{\mdef}[1]
{{\color{black}\sl\textbf{#1}}}

\renewcommand{\i}{\item}
\def\weight{\mbox{\sl weight}}

\def\lra{{\leftrightarrow}}
\def\M{{\mathcal{M}}}
\def\T{{\mathcal{T}}}
\def\G{{\mathcal{G}}}
\def\R{{\mathcal{R}}}
\def\F{{\mathcal{F}}}
\def\X{{\mathcal{X}}}

\def\BibTeX{{\rm B\kern-.05em{\sc i\kern-.025em b}\kern-.08em
    T\kern-.1667em\lower.7ex\hbox{E}\kern-.125emX}}

\begin{document}

\title{Data Inference from Encrypted  Databases: A Multi-dimensional Order-Preserving Matching Approach}

\author{\IEEEauthorblockN{Yanjun Pan}
\IEEEauthorblockA{\textit{University of Arizona} \\
Tucson, USA \\
yanjunpan@email.arizona.edu}
\and
\IEEEauthorblockN{Alon Efrat}
\IEEEauthorblockA{\textit{University of Arizona} \\
Tucson, USA \\
alon@cs.arizona.edu}
\and
\IEEEauthorblockN{Ming Li}
\IEEEauthorblockA{\textit{University of Arizona} \\
Tucson, USA \\
lim@email.arizona.edu}
\and
\IEEEauthorblockN{Boyang Wang}
\IEEEauthorblockA{\textit{University of Cincinnati}\\
Cincinnati, USA \\
boyang.wang@uc.edu}
\and
\IEEEauthorblockN{Hanyu Quan}
\IEEEauthorblockA{\textit{Huaqiao University}\\
Xiamen, China \\
quanhanyu@hqu.edu.cn}
\and
\IEEEauthorblockN{Joseph Mitchell}
\IEEEauthorblockA{\textit{Stony Brook University}\\
Stony Brook, USA \\
joseph.mitchell@stonybrook.edu}
\and
\IEEEauthorblockN{Jie Gao}
\IEEEauthorblockA{\textit{Stony Brook University}\\
Stony Brook, USA \\
jgao@cs.sunysb.edu}
\and
\IEEEauthorblockN{Esther Arkin}
\IEEEauthorblockA{\textit{Stony Brook University}\\
Stony Brook, USA \\
estie@ams.stonybrook.edu}
}

\maketitle
\begin{abstract}
Due to increasing concerns of data privacy, databases are being encrypted before they are stored on an untrusted server. To enable search operations on the encrypted data, searchable encryption techniques have been proposed. Representative schemes use order-preserving encryption (OPE) for supporting efficient Boolean queries on encrypted databases. Yet, recent works showed the possibility of inferring plaintext data from OPE-encrypted databases, merely using the order-preserving constraints, or combined with an auxiliary plaintext dataset with similar frequency distribution. So far, the effectiveness of such attacks is limited to single-dimensional dense data  (most values from the domain are encrypted), but it remains challenging to achieve it on high-dimensional datasets (e.g., spatial data) which are often sparse in nature. In this paper, for the first time, we study data inference attacks on multi-dimensional encrypted databases (with 2-D as a special case). We formulate it as a 2-D order-preserving matching problem and explore both unweighted and weighted cases, where the former maximizes the number of points matched using only order information and the latter further considers points with similar frequencies. 
We prove that the problem is NP-hard, and then propose a greedy algorithm, along with a polynomial-time algorithm with approximation guarantees. Experimental results on synthetic and real-world datasets show that the data recovery rate is significantly enhanced compared with the previous 1-D matching algorithm.
\end{abstract}

\section{Introduction}
Data outsourcing has become popular in recent years. Small businesses or individual users choose to delegate their data storage to public cloud servers (such as Amazon EC2 or Google Cloud) to save operational costs. Meanwhile, data breaches happen at an increasing rate, which compromise users' privacy.
For instance, the Yahoo! data breaches reported in 2016 affected 3 billion user accounts \cite{yahoo}.
This is exacerbated by recent scandals of data misuse (such as the Facebook-Cambridge Analytica case \cite{facebook}),
which increases the level of distrust from users. To address this issue, end-to-end encryption is commonly adopted to encrypt the data before it is uploaded and stored on an untrusted server. In order to enable efficient utilization over encrypted data (such as answering queries), many cryptographic techniques called searchable encryption (SE) \cite{song2000practical,bellare2007deterministic,curtmola2011searchable} have been proposed. The main challenge for SE is to simultaneously provide flexible search functionality, high security assurance, and efficiency. Among existing SE schemes, Order-Preserving Encryption (OPE) \cite{boldyreva2009order,boldyreva2011order,popa2013ideal,kerschbaum2014optimal} has gained wide attention in the literature due to its high efficiency and functionality.  In particular, OPE uses symmetric key cryptography and preserves the numeric order of plaintext after encryption, which supports most Boolean queries such as range query. Well-known systems for encrypted database search using OPE include: CryptDB \cite{popa2011cryptdb}, Google Encrypted Bigquery Client \cite{Bigquery}, and Microsoft Always Encrypted Database \cite{Microsoft}. 

Many early OPE schemes, unfortunately, were shown to leak more information beyond what is   necessary (i.e., the order between plaintexts). Therefore, schemes that satisfy ideal security guarantees (that only the order is leaked) have been proposed \cite{popa2013ideal, kerschbaum2014optimal}. However,  recent research \cite{NKW15, GSBNR17} showed that it is possible to infer/recover a significant portion of plaintexts from their OPE ciphertext, using only the ciphertext order relationships, as well as some auxiliary dataset with data frequencies similar to a target dataset. For example, Naveed et al. \cite{NKW15} attacked an encrypted medical database where users' age column is encrypted using OPE. Later, the attack was improved by Grubb et al. \cite{GSBNR17},   with an additional restriction of non-crossing property in the   matching algorithm.  

We note that, to date, all the successful inference attacks against OPE are limited to one-dimensional data \cite{NKW15, GSBNR17}. That is, even though a database may have multiple numeric columns/dimensions, where each of them being encrypted by OPE, each of these columns are treated separately when they are matched with plaintext values. This works well for \textit{dense} data, i.e., where most of the values from the whole data domain have corresponding ciphertexts present in the database, such as age \cite{GSBNR17}. Intuitively, the denser the data is, the more effective the attack is, because more constraints imposed by the ciphertext order  reduces the uncertainty of their corresponding plaintext values. 
However, for multi-dimensional databases, applying such 1-D matching algorithms on each dimension separately can yield results far from  optimal, since it neglects that for each pair of data tuples the order-preserving constraints on all the dimensions must be held jointly, leading to a much larger search space than the actual one and therefore more ambiguity in matching.  
%
In addition, 
for higher dimensional data (such as spatial/location data), the data tuple tends to be increasingly sparsely distributed in the domain, which  invalidates the one-dimensional matching approach (unless  the ciphertext and known plaintext datasets are highly similar with each other). Therefore, we wonder whether it is still feasible to recover OPE-encrypted data tuples for multi-dimensional, sparse databases? This turns out to be a very challenging problem.

In this paper, we study data inference attacks against multi-dimensional encrypted databases by jointly considering all the dimensions and leveraging only the ciphertext tuples' order and frequency information, with the help of an auxiliary plaintext dataset with similar frequencies (the same assumption is adopted by many  previous works). We formulate the \emph{order-preserving matching problem} first in 2D but later extend it to 3D and higher dimensions. 
In the unweighted case, given 
an OPE-encrypted database and an auxiliary plaintext dataset, each containing a set of points in 2D, we maximize the number of points in a matching from the ciphertext to the plaintext, where order-preserving property must be simultaneously satisfied in both dimensions. Such a matching is called a non-conflicting matching in which the $x$/$y$ projection of one edge in the matching cannot contain the projection of another edge in the matching. In general we also consider point frequency (the number of records with the same value), points matched with a smaller frequency difference are given higher weights and we maximize the total weights of the matching. 

We show that our problem can also be formulated as an integer programming problem (ILP), and prove its NP-hardness by reducing it to sub-permutation pattern matching problem.  Then we propose a greedy algorithm, along with an approximation algorithm with $O(n^{2.5} \log^3 n )$ runtime and an approximation factor of $O(\sqrt n)$.  This algorithm exploits the geometric structure of the problem, which is based on the idea of finding jointly heaviest monotone sequences (i.e., sequence of points with either increasing or decreasing order on each dimension) inside the auxiliary and target datasets. The main contributions of this paper are summarized as follows:

(1) To the best of our knowledge, we are the first to study data inference attacks against  multi-dimensional OPE-encrypted databases by jointly considering all the dimensions simultaneously. We formulate a 2-D order-preserving matching problem and show its NP-hardness. 

(2) We design two 2-D order-preserving matching algorithms, including  a greedy  and a polynomial time  algorithm  with approximation guarantees. We consider both unweighted and weighted cases, with different weight functions. We further explore efficiency enhancement using tree-based data structures. We also discuss extensions to higher dimensions. These algorithms have independent interest beyond the applications in this paper. 

(3) We evaluate the efficiency and data recovery rate of our algorithms over both synthetic and  real-world datasets for different application scenarios, including location-based services, census data, and medical data. Our results suggest that when the ciphertext dataset is highly similar to  a subset of the plaintext dataset, the greedy min-conflict algorithm performs the best; but, in general, when these two datasets have arbitrary intersections and are less similar, our monotone matching algorithm performs better. 
Overall, the recovery rate of our 2-D algorithms significantly outperform single-dimensional matching algorithms when the data is sparse in each dimension.


\section{Background and Related Work}



\subsection{Order-Preserving Encryption}  Order-Preserving Encryption (OPE) \cite{popa2013ideal} is a special encryption, where the order of ciphertexts is consistent with the order of
plaintexts. For instance, assume there are two plaintexts $(m_{1}, m_{2})$ and their OPE are
ciphertexts $(\llbracket m_{1} \rrbracket, \llbracket m_{2} \rrbracket)$, where $\llbracket m_{i}\rrbracket$ is the encrypted version of $m_i$ by following the common notations in previous studies \cite{popa2013ideal, GSBNR17}. 
If $m_{1} < m_{2}$, then $\llbracket m_{1} \rrbracket < \llbracket m_{2} \rrbracket$. With such property, comparison and sorting could be performed on encrypted data directly,
without the need to access plaintext. \alon{Worth explaining why is desired (e.g. performing queries on server side. How about "For example, the data is stored on the cloud, and the user could ask the server to perform (efficient) range counting or range reporting. The server would be able to perform these queries (efficiently) though the exact locations of these points is not revealed to the server.  }  While some OPEs are probabilistic and only reveal the order of data items \cite{kerschbaum2014optimal}, probabilistic OPEs increase the ciphertext size or require client-size storage, which scale poorly on sparse data. 
Most efficient OPEs are deterministic, and thus also reveal the frequency of data items \cite{popa2013ideal}.  In this paper, we focus on inference attacks on deterministic OPEs.  

\subsection{Inference Attacks on OPE via 1-D Matching}
While the security of OPEs has been proved formally under Ordered Chosen-Plaintext Attacks \cite{popa2013ideal}, several studies propose inference attacks to evaluate the privacy leakage of OPE ciphertexts. For instance, Naveed, et al. \cite{NKW15} proposed an inference attack, named cumulative attack, on 1-D OPE by leveraging frequency leakage only. The authors address the attack by running the Hungarian algorithm.
Grubbs et al. designed \cite{GSBNR17} leakage abuse attacks on 1-D OPE ciphertexts. The authors utilize both frequency and order leakage, and formulate the attack as 
a dynamic programming problem \cite{GSBNR17}. This leakage abuse attack performs faster than the cumulative attack and derives higher recovery rate. We briefly describe this leakage abuse attack below. 

Given an OPE-encrypted dataset $A=\{\llbracket a_{1} \rrbracket, \llbracket a_{2} \rrbracket, ..., \llbracket a_{n} \rrbracket\}$ and an unencrypted dataset $B=\{b_{1}, b_{2}, ..., b_{m}\}$ similar to $A$, an attacker tries to infer the plaintexts of $A$ without decrypting OPE ciphertexts, by leveraging the plaintexts of $B$ as well as the order and frequency information of $A$ and $B$.
Without loss of generality, the attack assumes that $A$ and $B$ are sorted, where $\llbracket a_{i} \rrbracket<\llbracket a_{j} \rrbracket$ for any $i<j$, and $b_{k}<b_{l}$ for any $k<l$. The attacker also assumes $n\leq{m}$. Let $F_{A}(\llbracket a \rrbracket)$ and $F_{B}(b)$ be the Cumulative Distribution Function (CDF) of the OPE ciphertexts of dataset $A$ and the plaintexts of dataset $B$ respectively. 
Now, construct a bipartite graph $H$ on vertex set $A$, $B$, in which the weight of an edge between  vertex $\llbracket a_{i} \rrbracket$ and vertex $b_{j}$ is defined as
\begin{equation*}
w(\llbracket a_{i} \rrbracket, b_{j}) = \kappa - |F_{A}(\llbracket a_{i} \rrbracket) - F_{B}(b_{j})|
\end{equation*}
where $\kappa$ is a pre-defined parameter and can be any integer greater than 1.

The attacker finds a max-weight bipartite matching in $H$ that is (one-dimensional) order-preserving (i.e., a vertex early in $A$ is mapped to an early vertex in $B$).
Intuitively, suppose we plot the points of $A$ and $B$ on two parallel lines in their order. If we draw the edges in the matching, these edges could not cross. That is, if $\llbracket a_{i} \rrbracket$ and $b_{j}$ are matched, any vertex in $\llbracket a_{k}\rrbracket$ with $k<i$ cannot be matched with vertex $b_{\ell}$ with $\ell >j$. Therefore, such a matching is also called a non-crossing matching.
The max-weight non-crossing matching can be found in time $O(mn)$ via dynamic programming.  
If vertex $b_{j}$ is matched with vertex $\llbracket a_{i} \rrbracket$, this attacker infers $b_{j}$ as the plaintext of OPE ciphertext $\llbracket a_{i} \rrbracket$.   


\subsection{Other Attacks on Encrypted Databases}

In addition to cumulative attacks and leakage abuse attacks, some other attacks have also been proposed against OPE. Durak et al.\ \cite{DDC16} proposed sort attacks on 2-D data encrypted by OPE. This attack performs a non-crossing matching on each dimension separately, and then improve the recovery results by evaluating inter-column correlation. Bindschaedler et al. \cite{BGCRS18} proposed an inference attack against property-preserving encryption on multi-dimensional data. This attack operates column by column. Specifically, it first recovers the column encrypted with the weakest encryption primitive, and then infers the next column encrypted by a stronger primitive by considering correlation. The attack is formulated as Bayesian inference problem. It also leverages \textit{record linkage} and \textit{machine learning} to infer columns that are strongly encrypted. In comparison, our proposed matching algorithms aim at optimally recover data tuples containing two or more dimensions as a whole. We utilize the order and frequencies of the 2-D tuples,  instead of single-dimension order and frequency in previous works.  In addition, we do not need explicit prior knowledge about the data correlations across dimensions within an encrypted dataset. 

Finally, reconstruction attacks \cite{KKNO16, LMP18} recover plaintexts on any searchable encryption that support range queries. Different from inference attacks, a reconstruction attack does not require a similar dataset as a reference but recover data based on access pattern leakage from a large number of range queries. However, reconstruction attacks often assume range queries are uniformly distributed, except \cite{GLMP19}, which is based on statistical learning theory. These works are orthogonal to this work.   


In this paper, we design two 2-D order-preserving matching algorithms that jointly consider the data ordering on 2D. We also extend the the 1-D matching algorithm in \cite{GSBNR17} to 2-D data for comparison. It turns out all the algorithms have advantages and limitations, as we describe in the evaluation and conclusion sections.


\section{Models and Objectives} 

\textbf{System Model.} In the system model, there are two entities, a \textit{client} and a \textit{server}. We assume that a client has a dataset (e.g., a location dataset) and needs to store it on the server. Due to privacy concerns, this client 
will encrypt the  dataset before outsourcing it to the server. 

We assume that the client encrypts the data using deterministic OPE, such that the server will be able to perform search operations (e.g., range queries)  over encrypted data without decryption. We assume that each dimension of the data is encrypted separately with OPE, such that search can be enabled for each dimension. The client's data set is denoted as $Q$ and its encrypted version as $\llbracket  Q\rrbracket$.

 \textbf{Threat Model.} We assume that the server is an \textit{honest-but-curious} attacker, who is interested in revealing the client's data but does not maliciously add, modify, or remove the client's data. In addition, we assume that the server is able to possess a similar dataset $P$  (in plaintext) as the client's dataset. In addition, we assume that $P$ and $Q$ have a significant common data points. 
 For those points in $Q$ that are also contained in $P$, they have similar frequency distributions.  For example, $Q$ can be the location data from Uber users, and $P$ can be a USGS spatial database ($Q$ can be considered to be randomly sampled from $P$). Or $P$ and $Q$ can be two location check-in datasets from two different social networking apps with partially overlapping locations.    
 
\textbf{Objectives}.
The attacker's goal is to   perform inference attacks  to maximally infer/recover the plaintext of encrypted database $\llbracket  Q\rrbracket$ without decryption, using only $\llbracket  Q\rrbracket$ and $P$ with the ciphertext/plaintext order, either with or without frequency of  points in both datasets. He aims at recovering the database points exactly. We define the \textit{recovery rate} as the primary metric to measure the privacy leakage of the inference attack.

\textbf{Recovery rate}: If an attacker infers $n$ points, $m'$ of which are correct inference (the same as their true plaintext points), 
then the recovery rate  is $m'/n$. In addition,   we   consider  both the \textit{unweighted} version of the above metrics, where each unique point/location is counted once, or the \textit{weighted} version where the frequency is considered as well (number of `copies' of the same point, e.g. the number of customers in a restaurant). The former can be regarded as ``point-level'' and the latter is ``record-level''. Intuitively, to maximize   the weighted recovery rate, the points with larger frequencies should be correctly matched with high priority. 




\section{2-D Order-Preserving Matching}

We formulate an  order-preserving matching problem in two dimensions. Let $P$ and $Q$ be two finite sets of points in the plane. $P=\{p_1, p_2, \cdots, p_n\}$ and $Q=\{q_1, q_2, \cdots, q_m\}$.
If $p\in P$ is matched to $q\in Q$, we denote it as an edge $(p, q)$ and sometimes also denoted as \textbf{}$p\lra q$.  
We say that a matching $\M$ between $P$ and $Q$ is {\em order preserving} if there exist two monotone functions $\psi, \phi$ such that if $(p,q)\in\M  $ (for $p\in P, q\in Q$) then $q.x=\psi(p.x)$, $q.y=\phi(p.y)$.  

It is convenient to consider an alternative, equivalent way to define order preserving, in terms of ``conflicts''. We say that two edges $(p,q)\in \M $ and $(p',q')\in \M$ are {\em in $x$-conflict with each other} if the $x$-projection (interval) of one edge contains the $x$-projection (interval) of the other edge; the notion of being in {\em $y$-conflict} is defined similarly.
We say that a matching $\M$ is a {\em non-conflicting matching} of $P$ and $Q$ if it does not contain any $x$-conflicting or $y$-conflicting pair of edges. 
From the definitions, it is easy to see that a matching $\M$ is order preserving if and only if it is a non-conflicting matching.
\vspace{-15pt}
\begin{figure}[htbp]
\begin{center}
 \includegraphics[scale=0.5]{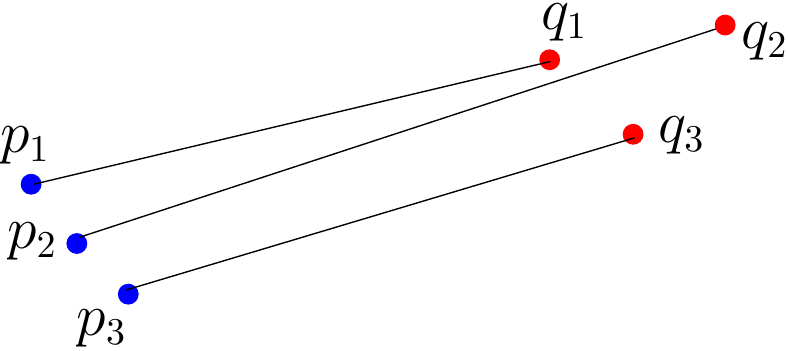} 
 \vspace{-10pt}
 \caption{\sl In this matching $\M$ between $P=\{p_1,p_2,p_3\}$ and $Q=\{q_1,q_2,q_3\}$, the edge $(p_2,q_2)$ is in $y$-conflict with edge $(p_1,q_1)$ and in $x$-conflict with edge $(p_3,q_3)$.
 \label{fig:3np}}
\end{center}
\vspace{-15pt}
\end{figure}

We say that a point $p'$ {\em dominates } $p''$, and write $p''\prec p'$,  if either (i) $p''.x< p'.x$, or (ii) $p''.x=  p'.x \mbox{~and~} p'.y<  p''.y$.
With this notation, two pairs $(p_i, q_j)$, $(p_{i'}, q_{j'})$ with $p_{i'}\prec p_i$ but $q_j \prec q_{j'}$ are in \emph{conflict}.


 
\subsection{Unweighted v.s. Weighted Version} \label{sec:weights_functions}

In this paper we study the problem of finding a maximum cardinality, or a maximum-weight  order preserving matching. 

In the unweighted version, we maximize the number of edges in a non-conflict matching between $P$ and $Q$. This formulation does not use information on data frequencies. 

To incorporate knowledge on data frequencies from $P$ and $Q$, we can define weight of matching a point in $P$ with a point in $Q$ and ask for the non-conflict matching with maximum weight. The goal is to minimize the total difference of the frequencies between each ciphertext and its matched plaintext points. Note that this may or may not be equivalent with the objective of maximizing the recovery rate. This depends on the similarity of the two datasets $P$ and $Q$: when the  frequencies of the same points are close in either dataset, max-weight matching will likely maximize recovery rate.

There are several possible choices of weight function. Assume $f(p_i)$, $f(q_j)$ are the frequencies of locations (resp.\ ) $p_i\in P , q_j\in Q$.  Then the weight of matching $p_i$ to $q_j$ could be one of the following weight function:   
\begin{enumerate}




\def\weight{{\mbox{\sl weight}}}

\item \label{weight1}
$ weight(p_i\lra q_j)= 
\min\{  f(p_i), f(q_i) \} $. 
The rational for this weight function is that if we consider  $f(p_i)$ and $f(q_j)$ as indicating the normalized number of items at point $p_i$ and $q_j$, then  $\min\{  f(p_i), f(q_i) \}$  indicates the maximum number of items could be matched.  


\item \label{weight2}

$\weight(p_i\lra q_j)=  \kappa- |f(p_j)-f(q_j)|$
, where $\kappa$ is a manually-picked  constant, usually as maximum of all $f(q_i)$ and $f(q_j).$
This is the cost function used in \cite{GSBNR17}. 
\end{enumerate}

\newcommand{\commentout}[1]{{}{\color{red}\tt something was here }}

\subsection{Integer Programming Formulation}

Given two sets of points, $P$ and $Q$,  we define a variable $x_{ij}$ that takes value $1$ if $p_i\lra q_j$ and $0$ otherwise. Now, we can formulate our matching problem as follows:
\begin{align*}
\mbox{\bf Maximize} 
&\sum_{i, j}  x_{ij} \cdot  w(p_i \lra q_j) \\
\mbox{\bf Subject to} &\sum_{j}x_{ij}\leq 1, \quad \forall i \\
&\sum_{i}x_{ij}\leq 1,\quad \forall j \\
& x_{ij}+x_{i'j'} \leq 1, \quad \forall (i,j), (i',j'),  \nonumber \\ 
 & \mbox{\textbf{s.\ t.\ } $(p_i,q_j)$ is in conflict with  $(p_{i'}, q_{j'})$.}
\end{align*}

The first two constraints imply that one point can only be matched to one other point. The last inequality is the non-conflicting constraint.
    

\def\Gconf{{G_{\mbox{\footnotesize Conf}}}}
\def\Econf{{E_{\mbox{\footnotesize Conf}}}}
\def\Gdom{{G_{\mbox{\footnotesize Dom}}}}

 

\def\MM{{MAX\_WEIGHT}}

\subsection{Related Results on Maximum Independent Sets}
Our problem can be phrased as a (weighted) maximum independent set (MIS) problem, in the conflict graph, defined below. 

\textit{Conflict Graph $\Gconf(P\times Q, \Econf)$}: 
the graph whose nodes are pairs of potentially matched points, one from $P$ and one from $Q$, and whose edges represent the conflict relationship: $(u,v)\in \Econf$ if the matched point pair $u\in P\times Q$ is in conflict with the matched point pair $v\in P\times Q$.

Unfortunately, this graph in our settings is enormous, and its node set has cardinality quadratic in the size of the input. Thus, pursuing our problem as a maximum independent set problem is likely impractical. 
In general, MIS has no polynomial-time constant factor approximation algorithm (unless $P = NP$); in fact, MIS, in general, is Poly-APX-complete, meaning it is as hard as any problem that can cannot be approximated within a polynomial factor \cite{BazganEP05}. However, there are efficient approximation algorithms for restricted classes of graphs.
In planar graphs, MIS can be approximated to within any approximation ratio $c < 1$ in polynomial time; MIS also has a polynomial-time approximation scheme in any family of graphs closed under taking minors~\cite{grohe2003local}.
In bounded degree graphs, effective approximation algorithms are known with approximation ratios that are constant for a fixed value of the maximum degree; for instance, a greedy algorithm that forms a maximal independent set by, at each step, choosing a minimum-degree vertex in the graph and removing its neighbors, achieves an approximation ratio of $(\Delta +2)/3$ on graphs with maximum degree $\Delta $ \cite{halldorsson1997greed}; hardness of approximation for such instances is also known~\cite{berman1999some},
and MIS on 3-regular 3-edge-colorable graphs is APX-complete~\cite{bazgan2005completeness}.

\section{NP-Hardness}
The problem of finding a maximum-cardinality order preserving matching (i.e., the unweighted case) is NP-hard. Therefore the weighted setting is also NP-hard.

We establish this by using a reduction from the problem {\sc Pattern Matching Problem for Permutations} (PMPP)~\cite{Bose1998}, which asks the following: Given a permutation $T=(t_1,t_2,\ldots,t_n)$ of the sequence $(1,2,\ldots,n)$ and a permutation $S=(s_1,s_2,\ldots,s_k)$ of the sequence $(1,2,\ldots,k)$, for $k\leq n$, determine if there exists a subsequence, $T'=(t_{i_1},t_{i_2},\ldots,t_{i_k})$, of $T$ of length $k$ (with $i_1<i_2<\cdots < i_k$) such that the elements of $T'$ are ordered according to the permutation $S$, i.e., such that $t_{i_j} < t_{i_{j'}}$ if and only if $s_{j} < s_{j'}$. We map a PMPP input pair of permutations, $(T,S)$, to a pair of points, $(P,Q)$, in the plane: Specifically, $P$ is the set $\{(i,t_i): 1\leq i\leq n\}$ of $n$ points corresponding to the permutation $T$, and $Q$ is the set $\{(i,s_i): 1\leq i\leq k\}$ of $k$ points corresponding to the permutation $S$. 
It now follows from the definition of an order preserving matching, and the specification of the PMPP, that there exists an order preserving matching of size $k$ between $P$ and $Q$ if and only if there is a subsequence $T'$ of $T$ of length $k$ such that the elements of $T'$ are ordered according to the permutation $S$. It follows that our (unweighted) order preserving matching problem is NP-hard.

\begin{theorem}
Given two point sets $P,Q \subset \reals ^2$, it is NP-complete to decide if there exists an order preserving matching $\M$ of cardinality $\min\{|P|,|Q|\}$ between $P$ and $Q$.    
\end{theorem}

\newcommand{\old}[1]{{}}

\old{For this, we utilize some results about pattern matching and permutations.  

A {\em permutation} $\sigma$ of $1,\ldots,n$ maps each integer $1\leq i\leq n$ to a unique index $1\leq \sigma(i)\leq n$. It is sometimes helpful to visualize this permutation in the form 
  $$\sigma=\begin{bmatrix}
1 & 2 & \cdots &n \\ 
\sigma_1 & \sigma_2 & \cdots& \sigma_n 
\end{bmatrix}$$

Since the symbols $P, Q$ and $T$ are used elsewhere in this paper, we will use the 'hat' notation exclusively to denote permutations. Bose et al.\ \cite{Bose1998} defined the 
{\em subpermutation pattern matching problem} as follows: We are given a permutation $\hat T = (\hat t_1,\hat t_2,.. , \hat t_n)$ of $1 \mbox{ to  }  n$, which we call the {\em text}, and a permutation 
$\hat Q = (\hat q_1 , \hat q_2, \ldots ,\hat \hat q_k )$  of $1,\ldots, k$, $k < n$, which we call the {\em pattern}. We wish to know whether there is a length $k$ subsequence of $T$, say $T’ = (\hat{t_{i_1}},\hat{t_{i_2}}, \hat{t_{i_k}}) $, with $ i_l < i_2 < i_k$   such that the elements of $\hat{T’}$ are ordered according to the permutation $\hat Q$ 
i.e., $ \hat{t_{i_r}} <\hat{ t_{i_s}}$, iff $\hat{p_r} < \hat{p_s}$. If $\hat{T}$ does contain such a subsequence, we will say
that $\hat{T}$ {\em contains}  $\hat{Q}$, or that $\hat{Q}$ {\em matches} into $\hat{T}$.

\begin{theorem}{\bf From Bose\cite{Bose1998}}
The subpermutation-pattern matching problem for permutations is NP-complete.
\end{theorem} 

\def\reals{\mathbb{R}}
\def\A{{\cal A}}



\def\P{{\cal P}}
We next show that the existence of a polynomial-time algorithm for the order-preserving matching problem would imply a polynomial-time algorithm for the subpermutations-matching problem. 
Given a permutation $\sigma,$ we create the set of points  $\P(\sigma)\subset\reals^2,$ as follows: $\{(i,\sigma(i)) ~|~ i=1..n\}$, i.e., assign to each value $i$ the point whose $x-$value is $i$ and whose $y-$value is $\sigma(i).$ 

Note that to identify a permutation $\sigma$, it is enough to specify for every $1\leq i < j \leq n$ whether $\sigma(i)< \sigma(j) $ or $\sigma(i)>\sigma(j)$. We need the following definition and result (from Bose et al.\ \cite{Bose1998})

 \begin{figure*}[htbp]
  \begin{center}
\includegraphics[scale=0.85]{2permutations2235.pdf} 
    \caption{Example of sets of points $\hat Q$ and $\hat T$ in $\reals^2$, representing the permutation $T$ and a permutation $Q$, which in this example, is a sub-permutation of $T$.       \label{fig1}}
     \end{center}
 \end{figure*}
 \def\inplane{{\subseteq \reals^2}}

Next, let $\hat Q$ be a permutation of $(1, \ldots, k)$ and $\hat T$ be a permutation of $(1, \ldots, n)$. Let $\hat Q = \mathcal{A}(Q) \subseteq \mathbb{R}^2$ and $T = \mathcal{A}(\hat T) \subseteq \mathbb{R}^2$ be the corresponding  point-sets which are the realization of these permutations, where $\mathcal{A}(\hat Q) = \{(i,\hat q_i), \hat q_i \in\hat  Q, i = 1, \ldots, k\}$, $\mathcal{A}(\hat T) = \{(i,\hat t_i), \hat t_i \in \hat T, i = 1, \ldots, n\}$. 

\begin{lemma} The following two conditions are equivalent 
\begin{enumerate}
    \item The permutation $\hat Q$ is a sub-permutation of $\hat T$.
    \item there is an \textbf{order-preserving} matching $f()$ that matches each point  $q_i\in \mathcal{A}(\hat Q)$ to a unique  point  $f(q_i)\in \mathcal{A}(\hat T)$.
\end{enumerate}
\end{lemma}

\begin{proof}
$\Rightarrow$  Assume that $\hat Q = (\hat q_1, \ldots, \hat q_k)$ is sub-permutation of $T = (\hat t_1, \ldots, \hat t_n)$, and let $i_1 < i_2 < \cdots < i_k$ be the index involved in the permutation in $\hat T$. By construction, we have:
\begin{align*}
    \mathcal{A}(\hat Q) & = \{(1,\hat q_1), \ldots, (k,\hat q_k)\}\\
    \mathcal{A}(\hat T) & = \{(1,\hat t_1), \ldots, (k,\hat t_n)\}
\end{align*}

Let $f$ be the matching that maps $(i,\hat q_i) \in \mathcal{A}(Q)$ to $(i_i,\hat t_{i_i}) \in \mathcal{A}(T)$. Then $f$ is an order-preserving matching: for any two matchings $(i,\hat q_i) \rightarrow (i_i,\hat t_{i_i})$ and $(j,\hat q_j) \rightarrow (i_j,\hat t_{i_j})$, with the definition of sub-permutation: 1) in $y$-axis: if $\hat q_i < \hat q_j$ then $\hat t_{i_i} < \hat t_{i_j}$ and vice versa, thus the order is preserved; 2) in $x$-axis:  if $i < j$, then $i_i < i_j$ and vice versa, similarly the order is preserved as well.

$\Leftarrow$ Assume that there is an order-preserving matching $f$ that matches each $(i,\hat q_i)\in \mathcal{A}(Q)$ to a unique  point  $(i_i,\hat t_{i_i})\in \mathcal{A}(\hat T)$. W.l.o.g., for any two matchings $(i,\hat q_i) \rightarrow (i_i,\hat t_{i_i})$ and $(j,\hat q_j) \rightarrow (i_j,\hat t_{i_j})$, we assume that $i < j$ and $\hat q_i > \hat q_j$. Then with the definition of order-preserving matching, we have $i_i < i_j$, $\hat t_{i_i} > \hat t_{i_j}$ correspondingly. Namely that the elements of $T' = (\hat t_{i_1}, \ldots, \hat t_{i_k})$ are ordered according to $\hat Q$, which is the definition of pattern matching for permutation problem itself. Thus, the permutation $\hat Q$ is a sub-permutation of $\hat T$ and with corresponding subsequence as $(\hat t_{i_1}, \ldots, \hat t_{i_k})$.
\end{proof}

Hence we have proven 

\begin{theorem}
Given two point sets $ T,  Q\subset \reals ^2$,  it is an NP-hard problem to determine whether there is an order preserving matching $\M$ that matches a point of $T$ to every point of $Q$.    
\end{theorem}

}  

\section{Algorithms }

\subsection{Greedy Minimum-Conflict Matching}\label{sec:gmv} 

In this heuristic, we create an order preserving  matching $\M\subseteq P\times Q$ in a greedy manner. We start with $\M$ empty, and at each iteration we add to $\M$ the edge that has the minimum number of conflicted edges among all potential future edges that could be selected. 
This heuristic is reminiscent of the minimum-degree heuristic of Halld{\'o}rsson and Radhakrishnan  \cite{halldorsson1997greed} that shows that similar heuristics provide a $(\Delta+2)/3$ approximation for finding a maximum independent set in graphs having maximum degree $\Delta$; 
however, in our setting, $\Delta$ might be $\Omega(|P||Q|)$, making this bound uninteresting. 

Formally, define for $p\in P, q\in Q$
\vspace{-8pt}
   \begin{align*}
   s(p,q)=\sum \{  w(p',q') \Big  | (p',q')   \mbox{ conflicts with } (p,q)  \mbox{ but } \nonumber \\ 
   \mbox{ not with edges currently in } \M \}  
\end{align*}    
and greedily select $(p^*,q^*)$ to minimize $s(p,q).$ A straightforward algorithm computes $s(p,q)$ directly (in time $O(n^2)$) for each of the $O(n^2)$ candidate edges $(p,q)$, in order to select each edge to be greedily added to $\M$. Overall, this is $O(n^4)$. 


\subsubsection{Unweighted Case}
Here, to expedite the algorithm to avoid the time $O(n^4)$ (per edge selected), we propose 
a weighted random sampling approach. 
We could find $(p^*,q^*)$ in amortized time $O(1)$ per pair $(p_i,q_j)$. 
This is done in two steps: We first compute for each $p_i$ the number $n^{\myarrow[135]}(p_i)$  of point $p\in P$ above and to the left of $p_i$.
\def\snearrow{{\footnotesize \nearrow}}
%
Similarly we define 
 $n^{\myarrow[45]}(p_i)$, $n^{\myarrow[-45]}(p_i)$, $n^{\myarrow}(p_i) $ and  $m^{\myarrow[45]}(q_j)$, $m^{\myarrow}(q_j)$, $m^{\myarrow}(q_j)$.  
Then the number of matching edges that are in conflict with $(p_i, q_j)$ can be computed by evaluating  the  products $n^\square(p_i)\cdot m^\square (q_j)  $,  where $\square$ is one of the 4 directions
$\myarrow[45],\myarrow[135],  \myarrow[-45],\myarrow[-135]. $ As easily observed, the number of conflicts is
\begin{align*}
s(p_i,q_j) & =
n^{\myarrow[135]}(p_i)m^{\myarrow[45]}(q_j  )+
n^{\myarrow[45]}(p_i) m^{\myarrow[135]}(q_j  ) \\
& \quad +n^{\myarrow[-135]}(p_i)m^{\myarrow[-45]}(q_j  )+
n^{\myarrow[-45]}(p_i) m^{\myarrow[-135]}(q_j  ) \\
& =nm-\sum_{\square \in \myarrow[135], \myarrow[45],\myarrow[-135],\myarrow[-45]} {m_i^\square} {n_i^\square} 
\end{align*}

%

%
We pick the edge minimizing this expression. Of course, once one edge is picked during the greedy matching algorithm, these numbers need to be recomputed, since multiple edges are not valid anymore.  

We note that after the first iteration, when partial matching $\M$ is not empty, the values of $n_i^\square, m_j^\square$ reflects only edges not violating edges of $\M.$ However, computing these values for every $p_i, q_j$ in time  $O(n^2)$ is straightforward. 

\subsubsection{Weighted case} We propose two basic methods.

\noindent{\bf Random sampling:   } 
 We consider all $n^2$ potential edges, $P\times Q$, compute the weight of each, and pick a random sample $R$ of (expected) size $k$, where the probability of picking $(p,q)$ is 
 $$k {\frac{{w(p,q)}}{{\sum_{p',q'} w(p',q')}}}$$  
 Next, we 
greedily find a min-violation edge, with the violation computed with respect to $R$ only. So the expected running time for this stage is $O(n^2k^2)$ per edge added to~$\M$.

This method can be enhanced further, for weight-function (\ref{weight1}), where the weight function is computed with respect to a random sample of {\em vertices} picked according to their weight. 

\vspace{2mm} 

\def\eps{{\varepsilon}}

\noindent{\bf $\eps$-approximation  via scaling algorithm}.   For the weight function  (\ref{weight1}) ($w(p_i,q_j)=\min\{f(p_i), f(q_j)\}$),  a faster approach is proposed.  

Let $w_{\min}, w_{\max}$ be the minimum weight and maximum weight. Consider the logarithmic number of
levels $ \{ w_{\min}(1+\eps)^i \}$ for every $i$ such  that $w_{\min}(1+\eps)^i \leq w_{\max}$. At the $i$'th step, we consider only the vertices with weight $\geq w_{\min}(1+\eps)^i, $ find the number $\zeta_i$ of edges conflicting  $(p,q)$  using the unweighted $O(n^2)$ algorithm, and sum the (rescaled) values $\sum (1+\eps)^i\zeta_i$ as an estimation of $s(p,q).$ It is easy to see that an (unscaled) edge conflicting $(p,q)$ will be counted once, with its weight error bounded by a factor of $(1+\eps). $ \alon{some grownup -- please check}  


\subsection{Greedy via Monotone Sequences   } 
Given $P$, we say that a sequence  $(p_1, p_2,,,p_k)$ is an \mdef{monotone increasing sequence} if  $p_i.x\leq p_{i+1}.x$ and $p_i.y\leq p_{i+1}.y$, for all $i$. Then, a subset $P'\subseteq P$ is said to be a monotone increasing subset if the sequence obtained by ordering $P'$ by $x$-coordinates is a monotone increasing sequence. Analogously we define sequences and subsets that are \mdef{monotone decreasing}.

In the previous section we discussed methods to greedily augment the matching by a single edge. 
One might wonder if it is possible, and whether it is more efficient, to add a collection of edges at each time. For example, Dynamic Programming proved useful in the 1-D case, and it is tempting to apply it for the 2-D case as well. However, applying similar techniques for the 2-D case seems very challenging. It is extremely hard to define sub-problems which are independent on each other, in the sense that the  solution of one does not depend on the solution to another. However, with a non-trivial hint on the approximation we obtain, we could define such a solution for monotone sequences.   Refer to Algorithm  \ref{alg:mono} for the pseudo-code. 

Essentially, if we opt to match $p\in P$ to $q\in Q$, then any decisions taken on the quadrant below and to the left of $p$ could not (in an order-preserving matching) affect matching in the quadrant {\em opposite} this quadrant, consisting of points above and to the right of $p$. Similar observations hold for every pair of opposite quadrants. This observation suggests our search for {\em monotone sequences.}  

Formally a sequence
$P^\uparrow=\{p_1\dots p_k\}\in P$ is an {\em increasingly monotone } sequence if $p_i\prec p_{i+1}.$ (for $i=1\dots k-1$).  Decreasing sequences are defined analogously.  Obviously if in a matching $\M$, $p_i$ is matched  to $q_i\in Q$ then the sequence 
$\{q_1\dots q_k\}$ is increasingly monotone as well. The {\em heaviest monotone sequence} is a monotone sequences maximizing the sum of weights of its edges. 
Given a partial matching $\M$, we describe in this section an algorithm that finds monotonic  sequences $P^\uparrow\subseteq P$ and $Q^\uparrow\subseteq Q$, a matching between them that does not conflict with $\M$ and is of maximum weight. We use this algorithm as follows: In iterations we find an optimal monotonically (increasing or decreasing) sequences with respect to $\M$, include the corresponding matched edges to $\M$ and continue.
Therefor we concentrate on efficient implementation of finding a single monotone matching. We discuss the case of monotonically increasing sequences. The case for monotonically decreasing sequences is handled analogously. 

\def\X{{\mathcal{X}}}
 
Let $\X =P\times  Q =\{(p_i, q_j) \big| p_i\in P, q_j\in Q\}$. By abusing notation, we also consider  each $(p_i, q_j)$ as a point in   $\reals ^4$,  with coordinates 
$(p_i.x, p_i.y, q_i.x, q_i.y)$.   We first describe the algorithm when $\X$ lies in $\reals^4$, and then show that we could orthogonally  project  $\X$ into $\reals^3$, and handle all querie as orthogonal three-dimensional range queries. For a point $p_i\in P$ we define $P_{\prec p_i}=\{p\in P \mid| p \mbox{ is dominated by } p_i\}.$ For $(p,q), (p',q')\in \X$ we say that $(p',q')$ dominates $(p,q)$, and write $(p,q)\prec (p',q')$ iff $p\prec p'$ and $q\prec q'.$ Similarly  for $ \mu=(p,q)$ (for $p\in P, q\in Q$), we write $w(\mu)$ to denote the weight of the matching edge $(p,q)$.  Fix $\mu=(p_i,q_j)$. We define $c[\mu]$ to be the maximum sum of weights of edges in any  maximum increasing monotone matching by using only  points of $P_{\prec p_i}$ to points of $Q_{\prec q_j}$ and ending at $\mu.$

\def\T{{\mathcal{T}}}

{  \footnotesize  

 

 \def\qinput#1{{#1}}
  \def\qcomment#1{{#1}}
  \def\qfor{{ FOR }}
    \def\qrof{{ ENDFOR }}
\def\qlet{{ LET }}
\def\qinput{{ QINPUT }}
\def\qfi{{ ENDIF }}
\def\qrepeat{{ REPEAT }}
\def\qif{{ QIF }}
\def\qor{{ OR }}
\def\qelse{{ ELSE }}
\def\qqif{{ IF }}
\def\qand{{ AND }}
\def\quntil{{ UNTIL }}
\def\qqthen{{ THEN }}
\def\qthen{{ THEN }}
\def\qdo{{ DO }}

}

\def\T{{\mathcal{T}}}
 
\def\ua{{^\uparrow}}

\def\G{{\mathcal{G}}}
 
    
    \def\R{{\mathcal{R}}}
    \def\F{{\mathcal{F}}}
    
    \def\TRUE{{\mbox{TRUE}}}
    \def\FALSE{{\mbox{FALSE}}}
 \def\X{{\mathcal{X}}}

 
 
 
 To obtain a fast asymptotic running time, we will use Algorithm \ref{alg:mono}. We maintain a 4-D orthogonal range tree $\T(\X)$ \cite{de1997computational}. Each leaf in the tree is associated with a node in $\X.$  
 Each internal  node $\eta\in T$ is associate with  
\begin{enumerate}
\item A range $R_\eta$ which is a rectangle in $\reals^4.$ 
\item  A subset  $\X_\eta\subseteq \X$ which include all points of $\X$ inside $R_\eta$. 
\i The point $\mu^*\in \X_\eta, $  which  is  the  last point  of  the   heaviest monotone sequence ending at $\mu^*$, for $\mu^*\in \X_{\eta}$.
 \item $c(\eta)$ --- the weight of this sequence.  
 \end{enumerate}

 The idea is to use an Orthogonal Range search data structure for the points in $\X$. We scan these points in topological  increasing order, so if $\mu\prec \mu'$ then we access $\mu'$ after accessing $\mu$. This will guarantee that $c(\mu)$ is fully computed at this point.

 \begin{lemma} \label{DS}  We could preprocess $\X$ into a data structure $\T$ such that the preprocessing time and space are both    $O(n^2\log^4 n)$,  given a query axis-parallel rectangle $R\subseteq \reals^4$, we could find  a set of $O(\log ^4n)$ nodes $\Xi=\{\eta_1\dots  \eta_k\}$ of $\T$, each corresponds to a subset $\X_{\eta_i}\subseteq \X$ that is  fully contained in $R$, and each  is associated with a value $c[\mu_i]$ which is $\max\{c[\mu']~|~ \mu'\in \X_{\eta_i}\}$ 
\end{lemma}

 Next we notice that filtering   points of $\X$ based on their first coordinate is not necessary. That is,   we only need to store each point of $(p,q)\in \X$ using only $(p.y,q.x,q.y)$, since a query on other regions yields that the result is zero, and will not effect the query time nor the correctness.
  
 \begin{lemma}
  So the data structure is in $\reals^3. $  Hence the query time of Lemma \ref{DS} is improved to $O(\log^3n).$  The space requirement also drops to $O(n^2\log^2n).$
\end{lemma}

\begin{proof}
 Since all weights are positive, and the points of $P$ are accessed in increasing $x$ order lexicographically, then once accessing $c((p_i,q_j))$, its value is strictly positive only due to a point $p'\in P$ such that $p'.x\leq p_i.x.$ Therefor there is no need to filter nodes of $\T$ based on their very first coordinate.  
\end{proof}

While orthogonal range trees are almost optimal theoretically, they suffer from several drawbacks. The space required is super-linear, and in practical applications, they tend to be inferior to other hierarchical spatial  data structures as kD-trees. The latter could be applied with linear memory, and faster search time on realistic data.

\begin{lemma} Instead of the orthogonal range tree, if we use a 3-dimensional kD-tree, the space requirement will be linear, while  the asymptotic running time per a query will increase to $O(|\X|  ^{(1-1/d)})= O((n^2)^{2/3})=O(n^{4/3}).$ 
\end{lemma}

    \newcommand{\minf}{{\scriptstyle-\infty}}


 

\begin{algorithm}[t]\label{alg1}
\small
\SetAlgoLined
 {\bf Input:}  $P$ and $Q$ (sorted in an increasing order), and a partial matching $\M\subseteq P\times Q$ \; 

 {\bf Init:}  Initialize the an orthogonal range search $\T$  for $\X$, \\and set $c(\eta)=0, \ \forall \eta\in\T.$\\
 \For{$i=1$ to $|P|$}{
 \For { $j=1$ to $|Q|$  }{ 
 \If{$\mu=(p_i, q_j)$ does not conflict any edge in $\M$} 
 {
 Set the range (rectangle) \\
$ R=\left\{  (\minf, p_i.y)\times 
  (-\infty, q_i.x )\times  (-\infty, 
  q_i.y )\right\}$\\
  
Perform a  range query in $\T$ with the range $R$ to obtain a set $\Xi=\{\eta_1\dots \eta_k\}$ of $O(\log^3n)$ nodes in $\T$\\
\tcc{Each $\eta_i\in\Xi $  corresponds to a region fully contains in $R$.} 
Let $\mu^*$ be $\arg\max\{c(\eta) | \eta\in \Xi\}$\\
Set $c(\mu)=w(p_i,q_j)+c(\mu^*)$ \\ 
\For{ each     $\eta'_i\in \T$ ancestor of $\mu$} 
{Set $c(\eta'_i)=max\{c(\eta'_i), c(\mu)\}$}
}
}
}
\caption{Finding heaviest increasing monotone chain \label{alg:mono} }
\end{algorithm}

\renewcommand{\minf}{\scalebox{1.75}{$-\infinity$}}

\def\C{{\cal C}}
\def\Q{{\cal Q}}

\subsubsection*{Running time and correctness } 
Given a partial matching $M,$  
it  takes $O(n^2\log^3n)$ to find the heaviest (max-weight) monotone matching not conflicting $M$. At this point these edges are added to $M$, and the process repeats. Since $P$ could be decompose into $\leq \sqrt n$  monotone sequences \cite{erdos1935combinatorial}, the number of iterations is $\leq \sqrt n. $  The overall running time is $O(n^{2.5}\log ^3 n),$ and the space is  $O( |P| |Q|\log^2n)$. Here $n=\max\{|P|,|Q|\}.$

\def\T{{\cal T }}

\subsection{Lower bounds} It is interesting to note that improving the bound below $\Omega(n^2)$ is unlikely, given that even if the points are on a line, then our problem is quite similar to the edit distance problem, and LCS problems,  for which recent lower bounds are proven under the SETH assumption \cite{backurs2015edit}. Hence we are only logarithmic distance away from the claimed optimum.





\subsection{Approximation guarantees  } 

\begin{lemma}
Let $opt$ be the maximum weight of the maximum order-preserving    matching. Then by using the monotone sequences algorithm \label{ZZZ}  we obtain a matching of weight   $\geq opt/\sqrt{\min\{|P|,|Q|\}}$. 
\end{lemma}
 \begin{proof} 
 Consider any partition of $P$ into monotone sequences. By the classical results of  Paul Erd{\"o}s   and  George Szekeres \cite{erdos1935combinatorial} $t\leq \sqrt n.$    Let $M^*$ be the set of all these edges, and let $M^*|_{P^\uparrow_i}$ be the set of edges with one endpoint in $P^\uparrow_i$ (for every $1\leq i \leq l$). 
 Since the sum of all matched edges in $M^*$ is $opt$, the sum of all edges $M^*|_{P^\uparrow_i}$ must be $\geq opt/l$. This cost is obviously not larger than the maximum we find along any monotone sequence. 
 \end{proof}

\subsection{Extension to Multi-Dimensional Data}   
Extending our matching algorithms to handle three or more dimensional data is well-motivated. Many databases have multiple columns of numerical data. For example, in the census database, there may be age, salary, zip code, etc. The data becomes increasingly sparse as the number of dimensions gets larger. Matching every single column separately will yield far from optimal results.

All our algorithms are still valid in this setting, where $P$ and $Q$ are points in $\reals^d$ for a constant $d$, $d\geq 3$.   However, the (worst-case) guarantees and running time of the monotone-greedy algorithm degrade. Since the longest monotone sequence of a set of $n$ points in $\reals^d$ is $\Theta(n^{1/d})$, we are only guaranteed a matching of weight $\geq opt/n^{1/d}$, where $opt$ is the optimum weight order-preserving matching.

\section{Evaluation}
In this section, we use both synthetic data and real-world datasets to evaluate the performance of our matching algorithms and attack effectiveness. 

\subsection{Data Sources} 

\noindent\textbf{Synthetic Data Generation}.
The synthetic data is generated in the following manner.
Let $R $ be a  set of uniformly distributed points in a 2-D area. Here $R$ is a superset of the data.  
Then, for every $r\in R$ assign a point weight $f'(r)$, 
which is a uniformly distributed pseudo-random integer on the interval $[f_{min}, f_{max}]$. Let $P$ denote the auxiliary/plaintext dataset and let $Q$ denote the target/ciphertext dataset. $P$ and $Q$ are generated as follows.

For the case of $Q$ being a subset of $P$: copy $R$ into $P$, and then generate $Q$ from $R$ with the following sampling process. For each $r\in R$, we randomly and independently copy it as a point $q\in Q$ with  probability  $\beta\in(0,1)$. And if $r$ was copied to a point  $q\in Q$, then the frequency $f'(q)$ is assigned to be a binomial variable  with probability $p_{\text{bion}}$ and expectation  $f'(r)\cdot p_{\text{bion}}$. For the other case that $P$ intersects with $Q$, we sample the points of both $P$ and $Q$ from $R$ randomly and independently with probability $\beta$, and similarly the weight of copied points are both sampled from a binomial distribution with probability $p_{\text{bion}}$. Finally, the integer point weights are normalized to frequency for $P$ and $Q$ in both cases, e.g. $f(p) = \frac{f'(p)}{\sum_{p'\in P}f'(p')}, \forall p \in P$.  

The rationale of point frequency following a binomial distribution is that, in the real-world, the set of people who appear in one dataset may choose to be present or not in another dataset  independently at random, e.g., Uber users can be regarded as randomly sampled from a USGS/census location database.

\smallskip\noindent\textbf{Real-world Data Sources}.
We use three real-world datasets to evaluate the performance of algorithms, and we start with the location check-in data from the social networking application Brightkite \cite{cho2011friendship}. The location coordinates are expressed in latitude and longitude. We extract the records in a certain area (latitude: [37.700887, 37.826664], longitude: [-122.512317, -122.386762]), and take the data collected from Apr. 2008 to Apr. 2009 as the auxiliary (unencrypted) dataset while the data from Sep. 2009 to Sep. 2010 as the target (encrypted) dataset. We randomly choose 500 and 300 points from the reference and target datasets as $P$ and $Q$ to perform the matching. Then we process the data in $P$ and $Q$ by discretizing it into location grids whose granularity is tunable, which means we reserve certain number of digits after the decimal point to represent that grid. For example, if we reserve three digits (0.001 as unit grid length which is approximately 0.1 km in the real-world), there are 399 and 247 points in auxiliary and target datasets respectively. Note that we do not need to actually encrypt the target dataset in our evaluation since it does not change the matching result, as we only use the order and frequency information of the points. 

We also use the city and town population totals for the US estimated by the United States Census Bureau \cite{uscensus}. To get the location information, we choose the 313 cities with latitude and longitude provided by Wikipedia \cite{wiki} (excluding Jurupa Valley), associating with 2010 census data as the auxiliary dataset $P$. And we sample 180 cities, combining with the corresponding 2018 estimated population data to get the target dataset $Q$. Each city record has two dimensional location data, namely $x$ and $y$ coordinates representing latitude and longitude, with its population as frequency. 


In addition to location datasets, we also evaluate our attacks in the context of medical data. We leverage a patient discharge dataset, which contains the distribution of inpatient discharges by principal diagnosis group for each California hospital \cite{medical}. We randomly sample the records collected in 2009 and 2014 as the auxiliary dataset and target dataset. More specifically, we select the numerical facility ID as the $x$ coordinate. For the $y$ coordinate, we convert the categorical diagnosis result to numerical data by assigning a specific diagnosis group to a numeric number. The frequency of each facility-diagnose pair is the recorded number of patients. There are overall 448 facilities in the original dataset, we first sample 30 of them as the facilities in $P$, and then sample 20 out of these 30 samples as facilities in $Q$. Then we extract the diagnosis and number of patients of the selected facilities to get $y$ coordinate and frequency. At the end, we obtain 425 records in auxiliary dataset $P$ and 272 records in target dataset $Q$.


\subsection{Evaluation Metrics}
For a given ciphertext record $c = \llbracket m \rrbracket$ and its plaintext $m$, the attack algorithm matches $c$ to a corresponding plaintext as $m'$ which might be different from  $m$. We say $c$ is correctly matched iff $m'=m$.
Note that, when $P$ intersects with $Q$, for a ciphertext $c = \llbracket m \rrbracket \in Q$, its plaintext $m$ may or may not be in the set $P$. Hence we define the record set of ciphertexts and plaintexts as $Q_r$ and $P_r$ by: $Q_r = \{c \in Q | c = \llbracket m \rrbracket,m \in P\}$, $P_r = \{m \in P | \llbracket m \rrbracket \in Q\}$. 
Define $I_{ij} = 1$ if point $p_i$ is correctly matched to point $q_j$, where $(p_i, q_j)$ is a plaintext-ciphertext pair; $I_{ij} = 0$ otherwise.  
We will use the following metrics for performance evaluation:
\begin{enumerate}
    \item Point recovery rate: $\sum_{j = 1}^{j = |Q|}I_{ij}  /|Q|$. 
    \item Record recovery rate (ratio of people in correctly matched locations):
$\sum\limits_{j = 1}^{j = |Q|}f_jI_{ij}/\sum\limits_{i = 1}^{i = |Q|}f_i$. 
    \item  Normalized objective, defined as the objective of an algorithm divided by the optimal ILP solution, 
    which is to evaluate how far the algorithm is from the optimal objective function in the ILP formulation.
\end{enumerate}
Note that there is an upper bound (the best one can do) of the recovery ratio. In the case of $Q$ is a subset of $P$, the upper bound is 1; in the case of $Q$ intersects with $P$, the upper bound is: $ |Q_r| $ / $|Q|$ for metric (1), and $\sum\limits_{i \in Q_r}f_i$ for metric (2). To remove the impact of the size of the intersection, we use \textit{normalized recovery rate} by dividing the recovery ratios with its upper bound in this paper.

Besides, we found that the data density and the similarity of the frequency distributions for 1-D and 2-D data are two factors that have significant impact on algorithm performance, we define the following metrics to quantify them:
\begin{enumerate}
    \item Overlap ratio: 
    the ratio of distinct records of ciphertext in $Q$ to plaintext in $P$. For 2-D data, it is $|Q_r|$ / $|P|$, for 1-D data, e.g. it is $|Q_{r_x}|$ / $|P_x|$ on x-axis, where $|Q_{r_x}|$ and $|P_x|$ are the number of unique x coordinates in $Q_r$ and $P$ respectively. If we 
    define the data density as the ratio of distinct points in $Q$ against total points in the domain of all possible plaintexts, then the overlap ratio can be regarded as the \textit{effective data density}, and when this ratio is small, the data is also sparse.
    
    \item Overall frequency similarity: 
    denote the union of points in $P$ and $Q$ as $R$. For a point $r$ in $R$, its frequency on $P$ is defined as the frequency of the  corresponding point if it appears in $P$, otherwise is 0. In this way, we get the frequency distribution of all these points on $P$ and $Q$ as $f(p) = \big(f(p_1),\cdots,f(p_k)\big)$ and $f(q) = \big(f(q_1),\cdots, f(q_k)\big)$. The overall frequency similarity of $P$ and $Q$ defined with Hellinger distance is:  
    $$H\big(f(p),f(q)\big) = \frac{1}{\sqrt{2}}\sqrt{\sum_{i=1}^{k}\Big(\sqrt{f(p_i)}-\sqrt{f(q_i)}\Big)^2}$$
    For Hellinger distance, we have $0 \leq H\big(f(p),f(q)\big) \leq 1$, and when $H\big(f(p),f(q)\big)$ is greater, the larger overall distance is between two distributions, and the less the similarity is.
\end{enumerate}







Also, we would like to compare the performance with a direct extension of the 1-D matching algorithm in \cite{GSBNR17} to 2-D data in the following way. We separately use the 1-D matching algorithm on each of the dimensions ($x$ or $y$ coordinates), with the order of $x$ (or $y$) and marginal frequency of each $x$ (or  $y$) coordinate. For a point $(c_x, c_y)$ in the target dataset, let $m_x$ and $m_y$ be the matched coordinates respectively, then we let $(m_x, m_y)$ as the matched point for $(c_x, c_y)$. Denote the point recovery rate of the 1-D matching on $x$ and $y$ coordinates as $r_x$ and $r_y$, then for the point recovery rate on 2-D data $r$ we have $r \leq \min \{r_x, r_y\}$.

 \subsection{Experimental Results for Synthetic Data}\label{sec:synthetic}
We implemented our proposed algorithms in Matlab, and all the experiments are run on a HPC cluster with 28 cores and 168GB memory. We first simulate with synthetic data, where we consider two cases: $Q$ is a subset of $P$, and $Q$ intersects with $P$. For each case, we generate $30$ different datasets, and under each dataset, we evaluate the algorithm performance with the two weight functions defined in Sec. 4.1.

\subsubsection{Synthetic Data, Case 1: $Q$ is a subset of $P$}
In this case, we set $P = R$, $|P| = 60$, $Q$ is obtained from $P$ with the data generation method described above, where $\beta = 0.6$, $p_{\text{bion}} = 0.7$. 
Results are shown in Fig. \ref{fig:synData} (a)-(b), where ``mix'' means we find both increasing and decreasing  monotonic matchings, and iteratively select the better one. We can see that the weight function 1 and 2 result in similar performance for all the 2-D algorithms, except that solution to ILP is the best (but also takes much more time, as shown in the right y axis). In this case the min-conflict algorithm outperforms all the other monotone sequences based algorithms under both weight functions. This is because the datasets are pretty similar.
In addition, our 2-D matching algorithms significantly outperform the extended 1-D algorithm (about 10\%) since our algorithms take into account the order of $x$ and $y$ axis simultaneously. 
\begin{figure}[!t]
	\begin{minipage}[b]{0.24\textwidth}
		\centering
		\includegraphics[width=\textwidth]{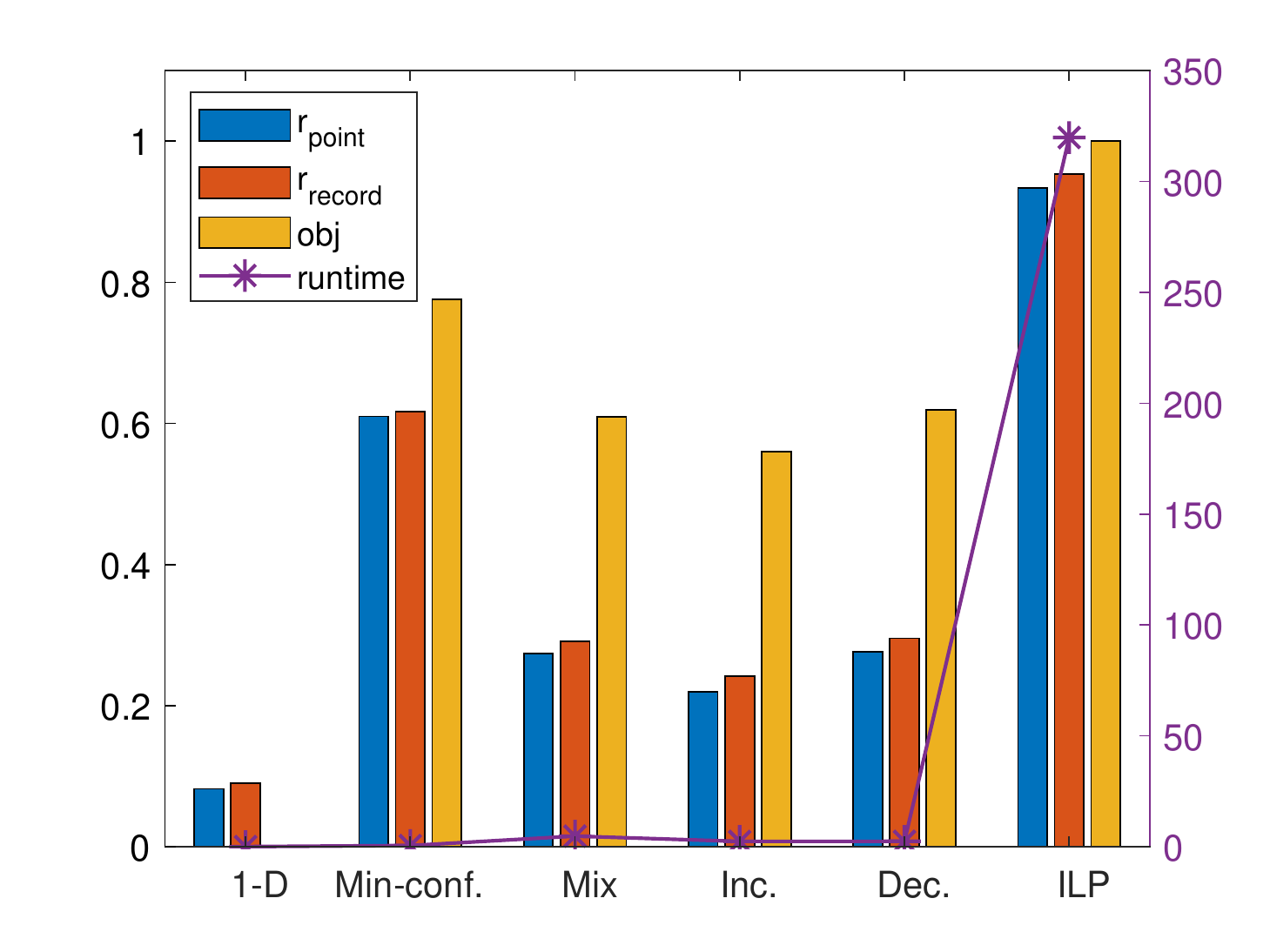}
		\\ \vspace{-1mm}
		{(a) Case 1, weight function 1}
	\end{minipage}
	\begin{minipage}[b]{0.24\textwidth}
		\centering
		\includegraphics[width=\textwidth]{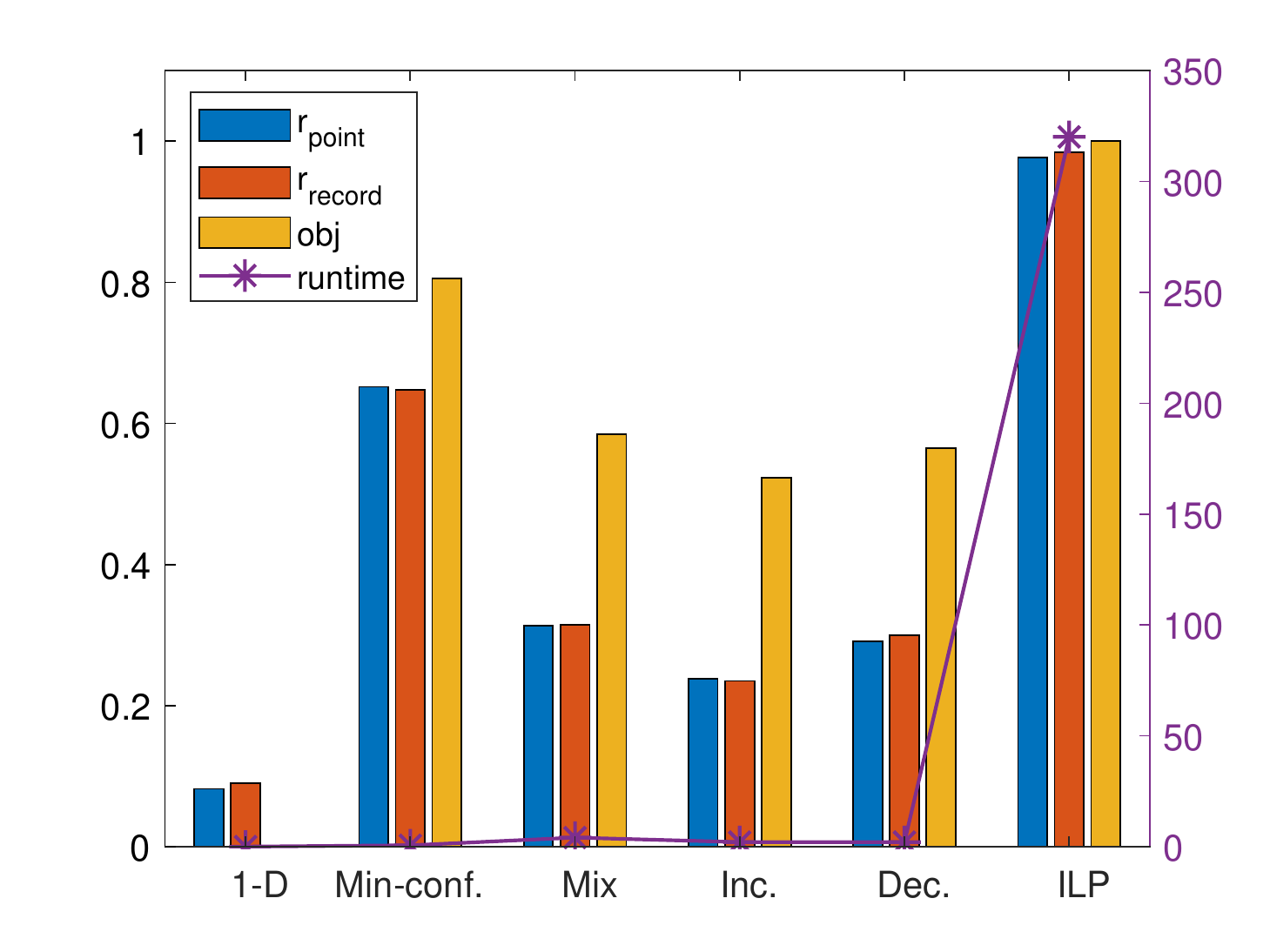}
		\\ \vspace{-1mm}
		{(b) Case 1, weight function 2}
	\end{minipage}
	\hfil
	\begin{minipage}[b]{0.24\textwidth}
		\centering
		\includegraphics[width=\textwidth]{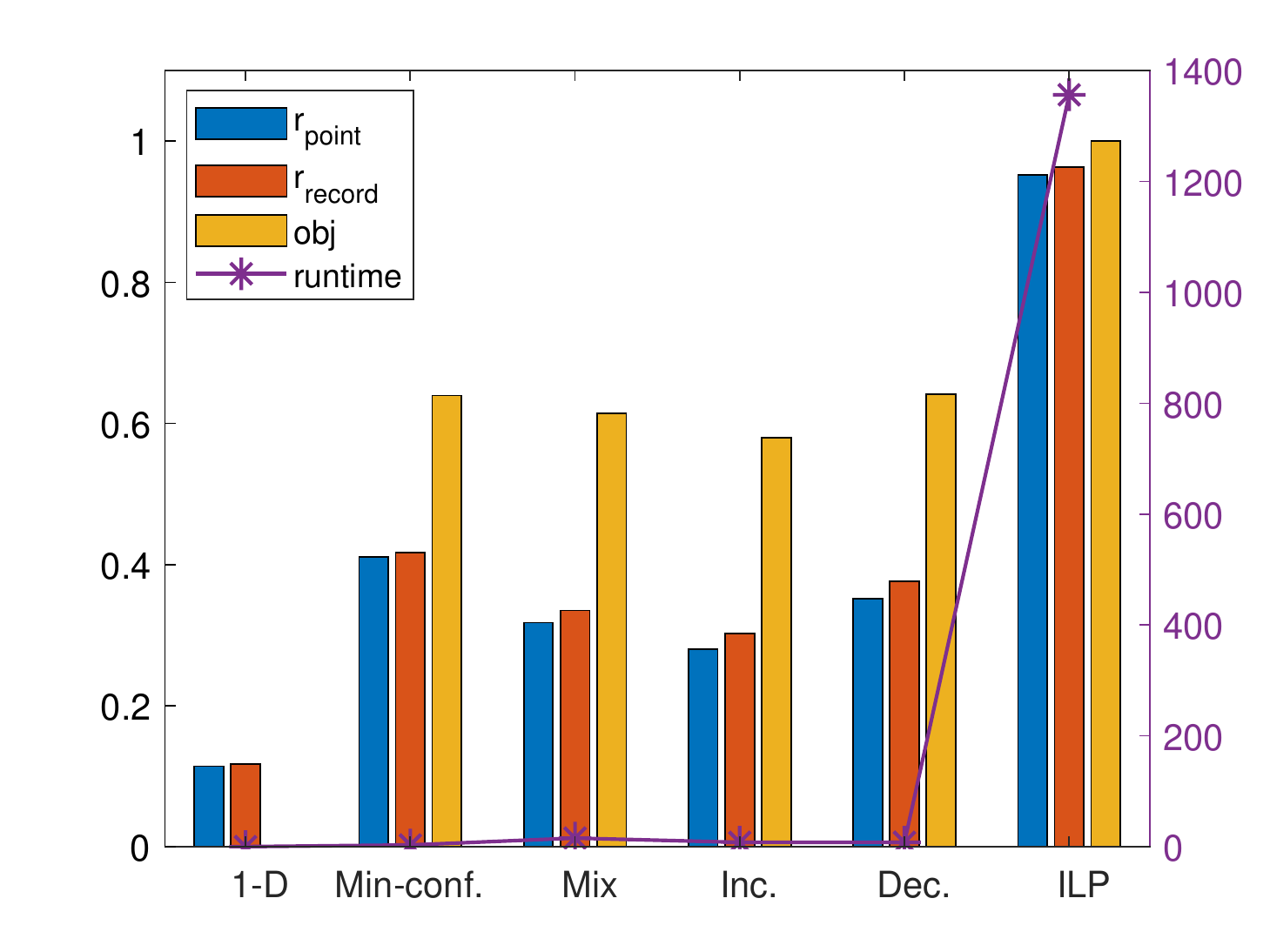}
		\\ \vspace{-1mm}
		{(c) Case 2, weight function 1}
	\end{minipage}
	\begin{minipage}[b]{0.24\textwidth}
		\centering
		\includegraphics[width=\textwidth]{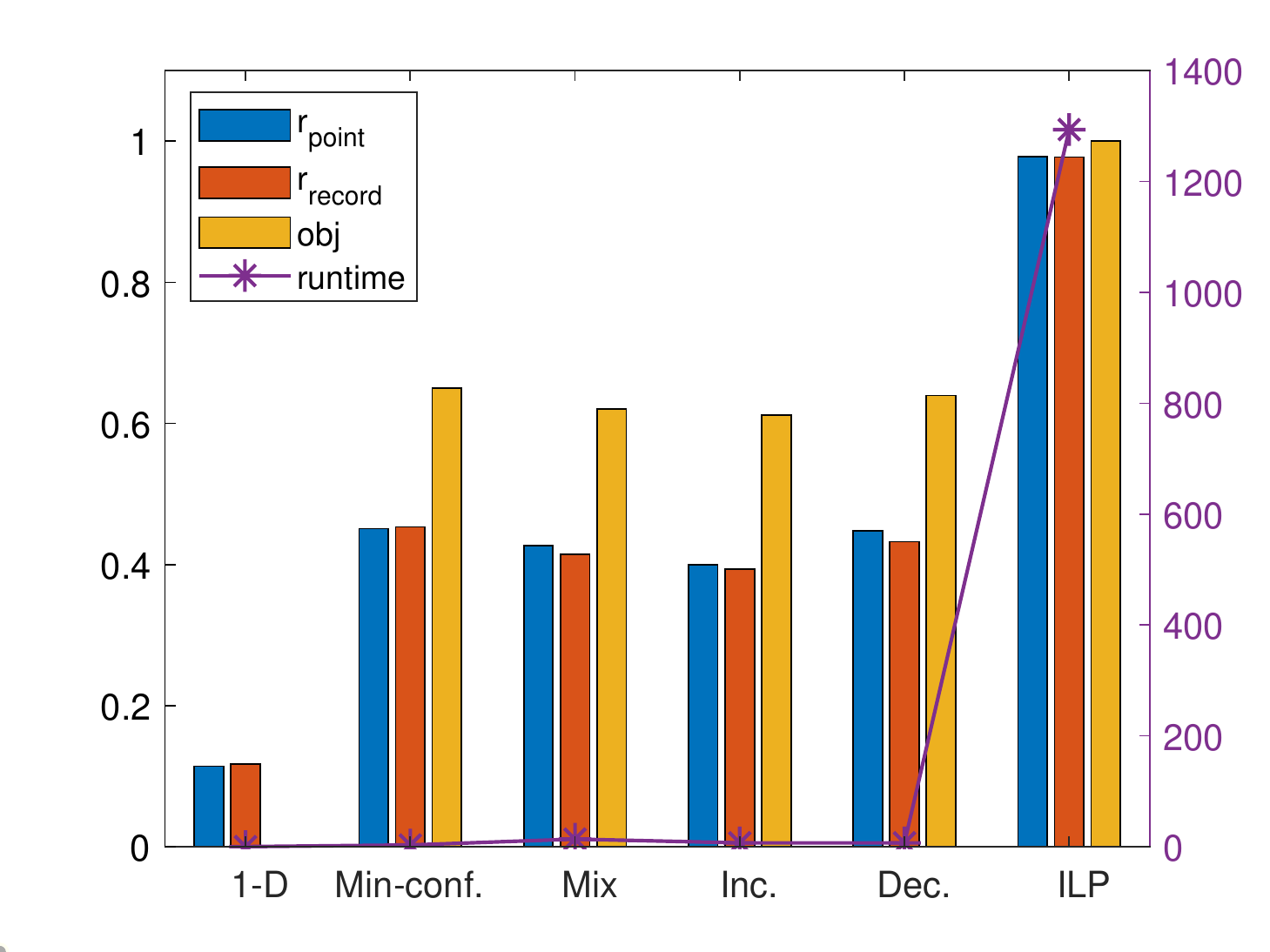}
		\\ \vspace{-1mm}
		{(d) Case 2, weight function 2}
	\end{minipage}
	\caption{Synthetic data: point recovery rate ($r_{\text{point}}$), record recovery rate ($r_{\text{record}}$) and objective (bar, left y axis), and running time in seconds (line, right y axis).}
	\label{fig:synData}	
\end{figure}

\subsubsection{Synthetic Data, Case 2: $P$ intersects with $Q$}
In this case, we set $|R| = 100$, and both $P$ and $Q$ are obtained from $R$ with the data generation method above, where $\beta$ and $p_{\text{bion}}$ are same as case 1. Besides,   about $60\%$ of points in $P$ and $Q$ are common. The average normalized recovery rates (point or record) of each algorithm among 30 runs are shown in Fig. \ref{fig:synData} (c)-(d). Still, both weight functions have similar performance.
And also the monotone-sequence-based algorithm and the min-conflict algorithm have comparable performance.
In addition, the normalized objectives are still close to that of the ILP solution ($>60$\%). Compared with case 1, in case 2 the normalized recovery rates are slightly lower for all the algorithms, since the two datasets are less similar and are more noisy.
 

\subsection{Experimental Results on Real-world Datasets}
We also evaluate our matching algorithms on the real-world datasets described at the beginning of this section. As the results in Sec. \ref{sec:synthetic} show that the two weight functions have comparable performance, we only apply the weight function \eqref{weight2} on real-world datasets. For the monotone based greedy algorithms, we use either increasing or decreasing order, but only show the best results among them and refer it as the monotone solution.

\subsubsection{Location-based Brightkite Check-ins}
The recovery rate of point and record with different matching algorithms are shown in Fig. \ref{fig:reald2} (a), we can see that the finer the granularity is, the lower the recovery rate we get. And when we have a coarse granularity, e.g. when 3 digits reserved, the points on 1-D are much denser, and also the frequency distributions are more similar (as shown in Table \ref{table:reald2}), so the extended 1-D algorithm outperforms our 2-D algorithms. However, as the granularity becomes finer, the points are sparser and the frequencies are less similar, then the monotone matching algorithm performs the best. Note  the lower recovery rate for finer granularity. This is because the algorithm relies on the   frequency difference between two points, and   both the ciphertext and plaintext's frequency distributions become more    homogeneous under higher granularity, which increases the ambiguity in weighted matching (inclined to be unweighted). The solution to ILP should give us the actual optimal result, but unfortunately, since there are too many constraints in ILP,   we cannot get its solution   in time.

\begin{figure}[!t]
	\begin{minipage}[b]{0.24\textwidth}
		\centering
		\includegraphics[width=\textwidth]{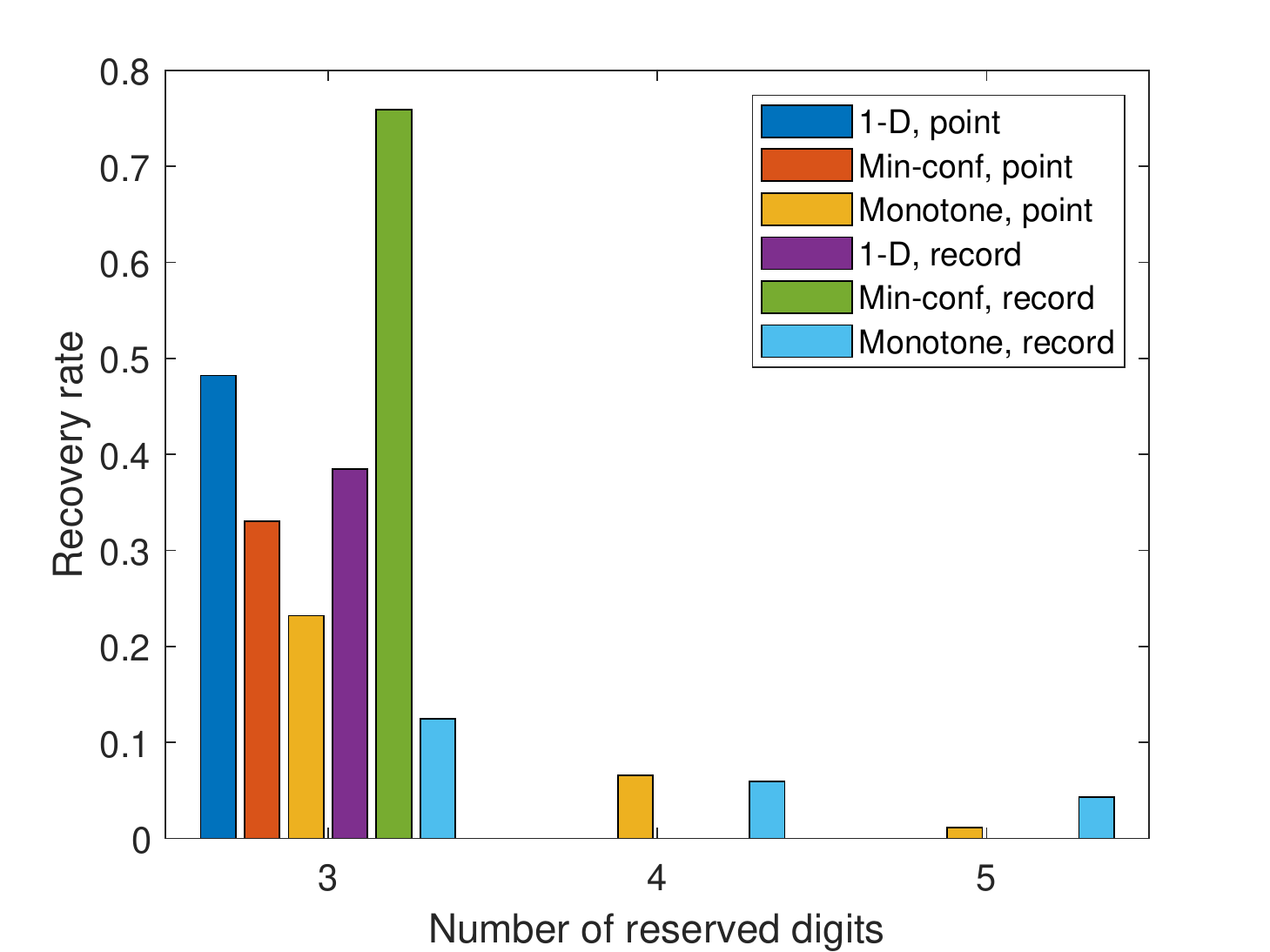}
		\\ \vspace{-1mm}
		{(a)}
	\end{minipage}
	\begin{minipage}[b]{0.24\textwidth}
		\centering
		\includegraphics[width=\textwidth]{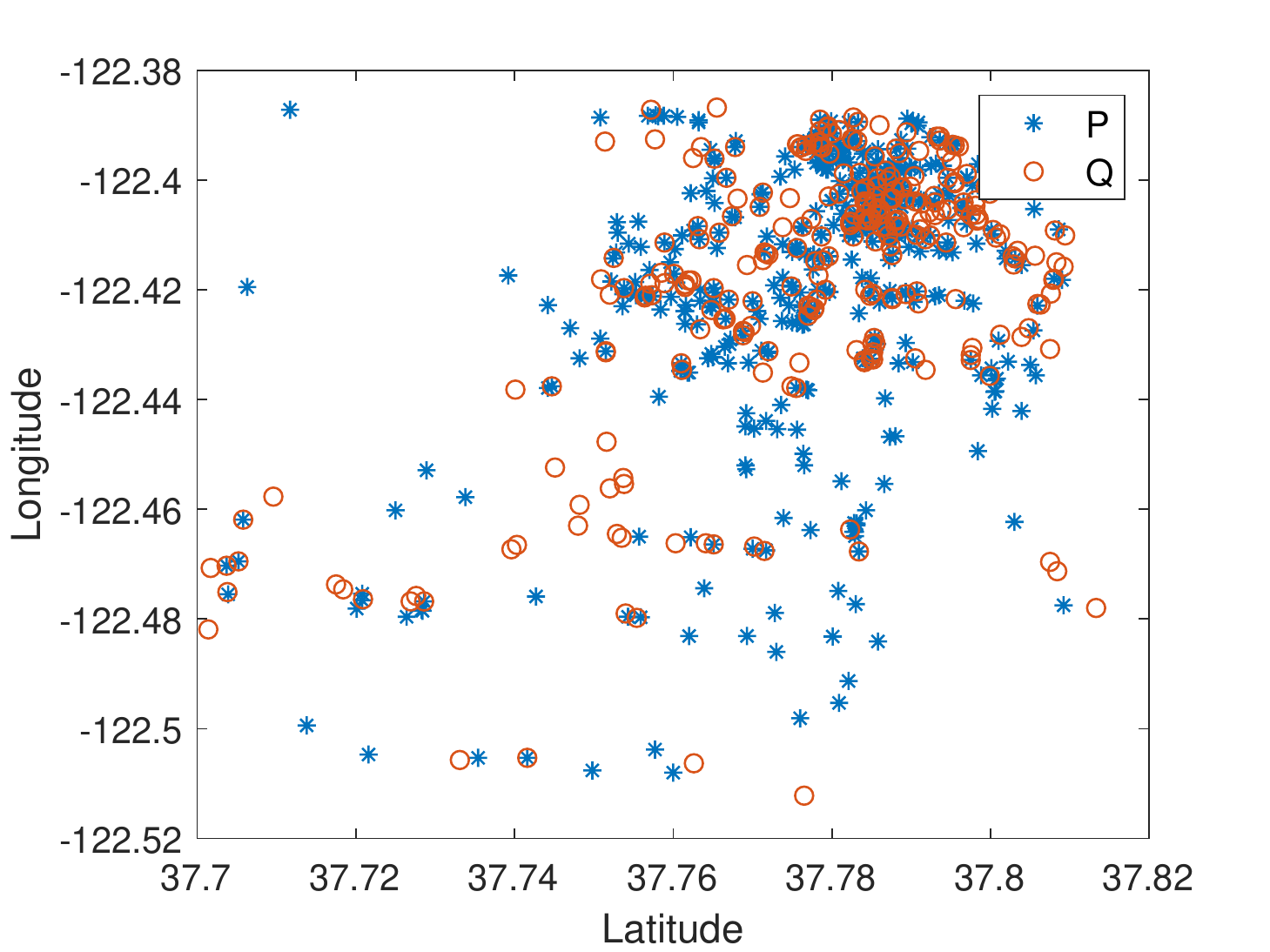}
		\\ \vspace{-1mm}
		{(b)}
	\end{minipage}
	\\ \vspace{-5mm}
	\caption{Brightkite: (a) point and record recovery rate; (b) distribution of locations when reserve 4 digits}
	\label{fig:reald2}	
	\vspace{-5mm}
\end{figure}

\begin{table}[!htb]
\caption{Effective data density and Hellinger distance for Brightkite}
\centering
\scalebox{0.8}{
\begin{tabular}{|c|*{7}{c|}}
\hline
\multirow{2}{*}{Reserved digits} &\multicolumn{3}{c|}{Effective data density} &\multicolumn{3}{c|}{Hellinger distance} \\
\cline{2-7}
  & $(x,y)$  & $x$ & $y$ & $(x,y)$ & $x$ & $y$\\
\hline
3 & 0.2807 & 0.8125 & 0.7283 & 0.5885 & 0.1786 & 0.2228 \\
\hline
4 & 0.1869 & 0.4595 & 0.4053 & 0.6393 & 0.6837 & 0.4773\\
\hline
5 & 0.1740 & 0.2120 & 0.2 & 0.6493 & 0.6214 & 0.6340\\
\hline
\end{tabular}}
\label{table:reald2}
\vspace{-10pt}
\end{table}

\subsubsection{City and Town Population} 
The matching results are given in Table \ref{table:reald1}, where the monotone matching yields a slightly lower recovery rate than the min-conflict algorithm, but better performance than the extended 1-D algorithm. Similarly, we compute the Hellinger distance (frequency similarity) and effective data density. It turns out that these two values on 2-D and 1-D data are the same, which are 0.5241 and 0.5751 respectively. This is because all the $x$ and $y$ coordinates are unique, so as to the $(x,y)$ data. Hence, in this case for the extended 1-D algorithm and 2-D algorithms, the effects of the data density and frequency distribution are the same, the recovery rate here is decided by the advantage of the algorithm itself. The 2-D algorithms outperform the extended 1-D algorithm on inference against 2-D OPE databases because we jointly consider the data orders on $x$ and $y$ coordinates. And we use an extreme case in Fig. \ref{fig:2D2D} to show the benefits of considering data orders in 2D.

\begin{table}[!htb]
\caption{Results for city and town population dataset}
\centering
\scalebox{0.8}{
\begin{tabular}{|c|*{5}{c|}}
\hline
\multirow{2}{*}{Algorithm}&\multicolumn{2}{c|}{Recovery rate} & \multirow{ 2}{*}{Runtime (s)} \\
\cline{2-3}
& Point  & Record & \\
\hline
1-D & 0 & 0 & 11.8830 \\
\hline
Min-conf. & 0.3778 & 0.3653 & 264.0576 \\
\hline
Monotone & 0.25 & 0.2631 & 309.6281\\
\hline
\end{tabular}}
\label{table:reald1}
\vspace{-10pt}
\end{table}

\subsubsection{Patient Discharge Data by Principal Diagnosis} We show the results of patient discharge data in Table \ref{table:reald3_1}, where  the monotone matching yields  very high recovery rates. But it is surprising that the min-conflict algorithm yields zero recovery rates, even worse than the extended 1-D algorithm. First we observe that for the adopted dataset, the set of $y$ coordinates of the auxiliary and target datasets are the same, which correspond to 19 principal diagnosis groups (hence the data density on $y$-axis is 1 in Table \ref{table:reald3_2}), and this results in an exact matching on $y$-axis for all the algorithms. Then the recovery rate is decided by the matching on $x$-axis, which means the 2-D order cannot benefit 2-D algorithms much. Hence in this case the noisy data frequency has a greater impact on the min-conflict algorithm, which leads  to poor performance.

\begin{table}[!htb]
\caption{Matching results for patient discharge data}
\centering
\scalebox{0.8}{
\begin{tabular}{|c|*{5}{c|}}
\hline
\multirow{2}{*}{Reserved digits}&\multicolumn{2}{c|}{Recovery rate} & \multirow{ 2}{*}{Runtime (s)} \\
\cline{2-3}
  & Point  & Record & \\
\hline
1-D & 0.0684 & 0.0290 & 0.249 \\
\hline
Min-conf. & 0 & 0 & 517.9117\\
\hline
Monotone & 0.8441 & 0.8995 & 196.3670\\
\hline
\end{tabular}}
\label{table:reald3_1}
\vspace{-15pt}
\end{table}

\begin{table}[!htb]
\caption{Effective data density and Hellinger distance for patient discharge data}
\centering
\scalebox{0.8}{
\begin{tabular}{|c|*{6}{c|}}
\hline
\multicolumn{3}{|c|}{Effective data density} &\multicolumn{3}{c|}{Hellinger distance} \\
\hline
$(x,y)$  & $x$ & $y$ & $(x,y)$ & $x$ & $y$\\
\hline
0.6188 & 0.6667 & 1 & 0.3798 & 0.1665 & 0.0960 \\
\hline
\end{tabular}}
\label{table:reald3_2}
\end{table}

 

\vspace{-15pt}
\section{Conclusion and Future Work}


\begin{figure}[!t]
\begin{center}
 \includegraphics[scale=0.45]{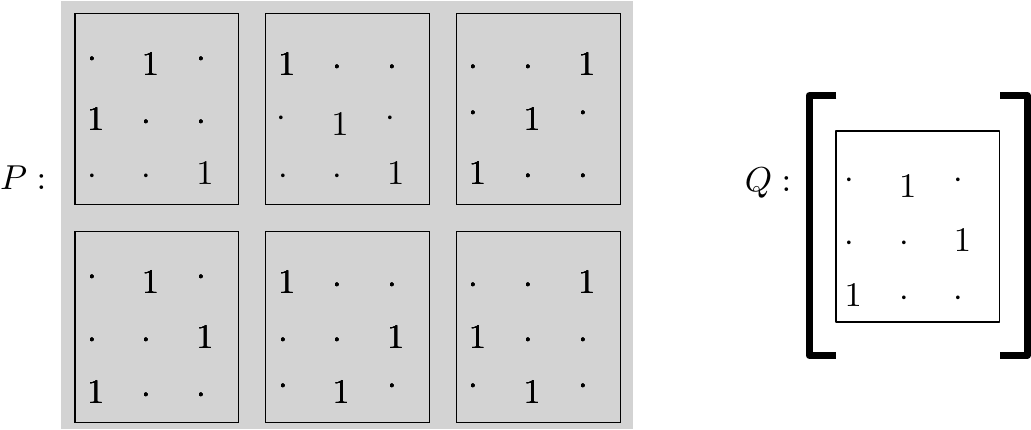} 
 \caption{\sl 1D-matching of \cite{GSBNR17} vs.\ 2D-matching.  Most points in $P$ have frequency $=0$, or do not appear in $Q$. These elements are represented by dots. Each block in $P$ is associated with a unique permutation of $(1, 2, 3 ).$  In $Q$, only one of these 6 blocks appear. Obviously the 2-D data uniquely indicates which blocks of $P$ is isomorphic to $Q$, and (if only blocks are considered), this block is revealed   as the  unique optimal solution to the  order-preserving matching.  On the other  hand, the sum of each row and column are identical, so no meaningful 1-D matching is possible. By replacing $3$ by an arbitrary $d$, we indicate that no less than $d!$ false matching are possible.  
 \label{fig:2D2D}}
\end{center}
\vspace{-15pt}
\end{figure}
In this paper, we studied the problem of  inferring  data from OPE-encrypted databases by jointly considering the multi-dimensional   order and   frequency of data tuples. 
We formulate it as a multi-dimensional matching problem, prove the NP-hardness, and propose   efficient  algorithms to solve it. Our algorithms exploit the geometric structure of the problem. 
We show that the monotone matching could be obtained in asymptotic time $O(n^{2.5}\log ^3n)$ that is comparable to the $O(n^2)$ 1-D algorithm. A simpler greedy algorithm is also provided. Experimental results on synthetic and real-world  datasets show that our algorithms perform  better than the extended 1-D algorithm when the data is sparse. In addition, the performance gain depends on data density and similarity of frequency distributions, but 2-D matching is more robust to noise (or dissimilarities between the frequency distributions). Our algorithms shed more light on the   security evaluation of OPE encrypted databases.

\vspace{2mm}
\noindent{\bf Extensions and Future Directions:}  Our experiments indicate that 1-D matching algorithm is likely to {\sl cluster} the data into meaningful regions, and  match region in an order-preserving way, using the aggregated data in each region. We think clustering could average/smooth the impact of noise, downsize the size of the problem, and circumvent issues when the datasets use different scale of coordinates, hence we would like to combine clustering with monotone matching algorithms. Clustering the  plaintext $P$ could be obtained based on  geographic/geometric vicinity in a straightforward way, thus we assume it is given in the rest for simplicity. However, finding corresponding regions in the ciphertext $Q$ is not obvious, as distances are distorted. Furthermore, finding these clusters cannot be decoupled from the matching process. We  have  obtained some preliminary results and we briefly describe them here. For any  two rectangles $R_1, R_2$ the cost of matching the point of $P$ inside $R_1$ to  points of $Q$ inside $R_2$ depends on the difference between their corresponding sum of frequencies in each rectangle.
Our  optimization function is to simultaneously obtain a clustering of $Q$ and an  one-to-one order-preserving matching  to $P$'s cluster. The 1-D problem can easily be solved in time  $O(n^2),$ using dynamic programming. This leads to an $O(n^5)$ algorithm for the 2-D problem for matching monotonic set of points in $P$ to a weakly-monotonic set of rectangles in $Q$.  We omit details due to lack of space.  Testing the effectiveness of these approaches is left for future work.

\bibliographystyle{IEEEtran}
\bibliography{refpaper}

\end{document}